\documentclass[preprint]{elsarticle}

\usepackage{amssymb}
\usepackage{graphicx}

\newtheorem{definition}{Definition}
\newtheorem{obs}{Observation}

\newtheorem{theorem}{Theorem}
\newtheorem{lemma}{Lemma}
\newproof{proof}{Proof}

\usepackage{algorithm}
\usepackage{algorithmic}
\newcommand{\alg}[1]{\mbox{\sf #1}}  

\newcommand{\comment}[1]{}

\usepackage{amsopn}

\journal{???}

\begin{document}

\begin{frontmatter}

\title{A Randomized Algorithm for Long Directed Cycle\tnoteref{title}}

\author{Meirav Zehavi\corref{cor}}
\ead{meizeh@cs.technion.ac.il}

\tnotetext[title]{{\em Abbreviations}: Long Directed Cycle (LDC).}
\cortext[cor]{Corresponding author.}

\address{Department of Computer Science, Technion IIT, Haifa 32000, Israel}

\begin{abstract}
Given a directed graph $G$ and a parameter $k$, the {\sc Long Directed Cycle (LDC)} problem asks whether $G$ contains a simple cycle on at least $k$ vertices, while the {\sc $k$-Path} problems asks whether $G$ contains a simple path on exactly $k$ vertices.  Given a deterministic (randomized) algorithm for {\sc $k$-Path} as a black box, which runs in time $t(G,k)$, we prove that {\sc LDC} can be solved in deterministic time $O^*(\max\{t(G,2k),4^{k+o(k)}\})$ (randomized time $O^*(\max\{t(G,2k),4^k\})$). In particular, we get that {\sc LDC} can be solved in randomized time $O^*(4^k)$.
\end{abstract}


\end{frontmatter}

\section{Introduction}

We study the {\sc Long Directed Cycle (LDC)} problem. Given a directed graph $G=(V,E)$ and a parameter $k$, it asks whether $G$ contains a simple cycle on {\em at least} $k$ vertices. At first glance, this problem seems quite different from the well-known {\sc $k$-Path} problem, which asks whether $G$ contains a simple path on {\em exactly} $k$ vertices: while {\sc $k$-Path} seeks a solution whose size is exactly $k$, the size of a solution to {\sc LDC} can be as large as $|V|$. Indeed, in the context of {\sc LDC}, Fomin {\em et al.}~\cite{representative} noted that ``color-coding, and other techniques applicable to {\sc $k$-Path} do not seem to work here.'' 

In this paper, we show that an algorithm for {\sc $k$-Path} can be used as a {\em black box} to solve {\sc LDC} efficiently. More precisely, suppose that we are given a deterministic (randomized) algorithm $ALG$ that uses $t(G,k)$ time and $s(G,k)$ space, and decides whether $G$ contains a simple path on {\em exactly} $k$ vertices directed from $v$ to $u$ for some given vertices $v,u\in V$.\footnote{Known algorithms for {\sc $k$-Path} handle the condition relating to the vertices $v$ and $u$.} Then, we prove that {\sc LDC} can be solved in deterministic time $O^*(\max\{t(G,2k),4^{k+o(k)}\})$ and $O^*(\max\{s(G,k),4^{k+o(k)}\})$ space (if $ALG$ is deterministic), or in randomized time $O^*(\max\{t(G,2k),4^k\})$ and $O^*(s(G,k))$ space (if $ALG$ is randomized).\footnote{The $O^*$ notation hides factors polynomial in the input size.} Somewhat surprisingly, we show that cases that cannot be efficiently handled by calling an algorithm for {\sc $k$-Path}, can be efficiently handled by merely using a combination of a simple partitioning step and BFS.

The first parameterized algorithm for {\sc LDC}, due to Gabow and Nie \cite{gabowNie08}, runs in time $O^*(k^{O(k)})$. Then, Fomin {\em et al.}~\cite{representative} gave a deterministic parameterized algorithm for {\sc LDC} that runs in time $O^*(8^{k+o(k)})$ using exponential-space. Recently, Fomin {\em et al.}~\cite{productFam} and the paper \cite{repFamUniApp} modified the algorithm in \cite{representative} to run in deterministic time $O^*(6.75^{k+o(k)})$ using exponential-space. It is known that {\sc $k$-Path} can be solved in randomized time $O^*(2^k)$ and polynomial-space \cite{williamsKPath}, and deterministic time $O^*(2.59606^k)$ and exponential-space \cite{mixing}. Thus, we immediately obtain that {\sc LDC} can be solved in randomized time $O^*(4^k)$ and polynomial-space, and deterministic time $O^*(6.73953^k)$ and exponential-space.

In the following sections, given a graph $G=(V,E)$ and a set $U\subseteq V$, we let $G[U]$ denote the subgraph of $G$ induced by $U$.

\section{Finding Large Solutions in Polynomial-Time}

We say that an instance $(G,k)$ of {\sc LDC} {\em seems difficult} if $G$ does not contain a directed cycle on $\ell$ vertices for any $\ell\in\{k,k+1,\ldots,2k\}$. Roughly speaking, given such an instance, we are forced to determine whether $G$ contains a {\em large} solution. This case, as noted in \cite{gabowNie08} and \cite{representative}, seems to be the core of difficulty of {\sc LDC}. We show, somewhat surprisingly, that under certain conditions, this case can be solved in polynomial-time. More precisely, this section proves the correctness of the following lemma.

\begin{lemma}\label{lemma:poly}
Let $(G,k)$ be instance of {\sc LDC}, and let $(L,R)$ be a partition of $V$. Then, there is a deterministic polynomial-time algorithm, \alg{PolyAlg}, which satisfies the following conditions.
\begin{itemize}
\item If $(G,k)$ seems difficult, and $G$ contains a simple cycle $v_1\rightarrow v_2\rightarrow\ldots\rightarrow v_t\rightarrow v_1$ such that $t > 2k$, $v_1,v_2,\ldots,v_k\in L$ and $v_{k+1},v_{k+2},\ldots,v_{2k}\in R$, \alg{PolyAlg} accepts. 
\item If $G$ does not contain a simple cycle on at least $k$ vertices, \alg{PolyAlg} rejects.
\end{itemize}
\end{lemma}

\begin{proof}
The pseudocode of \alg{PolyAlg} is given in Algorithm \ref{alg:poly}. Clearly, if the algorithm accepts, there exist two distinct vertices $v$ and $u$ such that $G$ contains two simple internally vertex disjoint paths, $P=(V_P,E_P)$ (from $v$ to $u$) and $P'=(V_P',E_P')$ (from $u$ to $v$), where $|V_P|=k$. In this case, $G$ contains a simple cycle, which consists of these paths, on at least $k$ vertices. Thus, the second item is correct.

\begin{algorithm}[!ht]
\caption{\alg{PolyAlg}($G=(V,E),k,L,R$)}\label{alg:poly}
\begin{algorithmic}[1]
\FORALL{$v\in L$ and $u\in L\setminus\{v\}$}\label{step:ite}
	\STATE Use BFS to find a simple path $P=(V_P,E_P)$ from $v$ to $u$ in $G[L]$ that minimizes $|V_P|$.
	\IF{$|V_P|\neq k$ or the path $P$ does not exist}\label{step:if1}
	 	\item Skip the rest of this iteration.
	 \ENDIF
	\STATE Use BFS to find a simple path $P'=(V_P',E_P')$ from $u$ to $v$ in $G[V\setminus (V_P\setminus\{v,u\})]$ that minimizes $|V_P'|$.
	\IF{the path $P'$ exists}\label{step:if2}
		\STATE Accept.
	\ENDIF
\ENDFOR
\STATE Reject.
\end{algorithmic}
\end{algorithm}

Now, we turn to prove the first item. To this end, suppose that the condition of this item is true. Then, we can let $C=v_1\rightarrow v_2\rightarrow\ldots\rightarrow v_t\rightarrow v_1$ be a simple cycle in $G$ such that $t > 2k$, $v_1,v_2,\ldots,v_k\in L$ and $v_{k+1},v_{k+2},\ldots,v_{2k}\in R$, {\em which minimizes} $t$. We need the following observations.

\begin{obs}
The number of vertices on the shortest path from $v_1$ to $v_k$ in $G[L]$ is exactly $k$.
\end{obs}

\begin{proof}
The existence of $C$ implies that we can let $P=(V_P,E_P)$ denote a path from $v_1$ to $v_k$ in $G[L]$ that minimizes $|V_P|$, and that we can assume that $|V_P|\leq k$. We furhter denote $P=u_1\rightarrow u_2\rightarrow\ldots\rightarrow u_{|V_P|}$, where $u_1=v_1$ and $u_{|V_P|}=v_k$. It remains to show that $|V_P|=k$. Suppose, by way of contradiction, that $|V_P| < k$. Let $v_i$ be the first vertex on the path $v_{k+1}\rightarrow v_{k+2}\rightarrow\ldots\rightarrow v_t\rightarrow v_1$ that belongs to $V_P$. Then, we can define a simple cycle $C'$ in $G$ as follows.
\begin{itemize}
\item If $i=1$: $C'=v_{k+1}\rightarrow v_{k+2}\rightarrow\ldots\rightarrow v_t\rightarrow (v_1=u_1)\rightarrow u_2\rightarrow\ldots\rightarrow (u_{|V_P|}=v_k)\rightarrow v_{k+1}$.
\item Else: Let $j$ be the index such that $v_i=u_j$. Then, $C'=v_{k+1}\rightarrow v_{k+2}\rightarrow\ldots\rightarrow v_{i-1}\rightarrow (v_i=u_j)\rightarrow u_{j+1}\rightarrow\ldots\rightarrow (u_{|V_P|}=v_k)\rightarrow v_{k+1}$.
\end{itemize}

Clearly, the number of vertices of $C'$ is smaller than $t$. Therefore, by the choice of $C$ and since $(G,k)$ is a seemingly difficult instance of {\sc LDC}, we have that $C'$ is a cycle on less than $k$ vertices. However, since $V_P\subseteq L$ and $v_{k+1},v_{k+2},\ldots,v_{2k}\in R$ (where $R=V\setminus L$), we have that $2k<i$. This implies that $C'$ is a cycle on at least $k$ vertices, and thus we have reached a contradiction.\qed
\end{proof}

\begin{obs}
Let $P=(V_P,E_P)$ be a simple path from $v_1$ to $v_k$ in $G[L]$ such that $|V_P|=k$. Then, $G[V\setminus (V_P\setminus\{v_1,v_k\})]$ contains a path from $v_k$ to $v_1$.
\end{obs}

\begin{proof}
Denote $P=u_1\rightarrow u_2\rightarrow\ldots\rightarrow u_k$, where $u_1=v_1$ and $u_k=v_k$. If $V_P\cap \{v_{k+1},v_{k+2},\ldots,v_t\}=\emptyset$, then the claim is clearly true, since then $v_k\rightarrow v_{k+1}\rightarrow\ldots\rightarrow v_t\rightarrow v_1$ is a path in $G[V\setminus (V_P\setminus\{v_1,v_k\})]$. Suppose, by way of contradiction, that $V_P\cap \{v_{k+1},v_{k+2},\ldots,v_t\}\neq\emptyset$. Then, we can let $v_i$ be the first vertex on the path $v_{k+1}\rightarrow v_{k+2}\rightarrow\ldots\rightarrow v_t$ that belongs to $V_P$. Let $j$ be the index such that $v_i=u_j$. We have that $C'=v_{k+1}\rightarrow v_{k+2}\rightarrow\ldots\rightarrow v_{i-1}\rightarrow (v_i=u_j)\rightarrow u_{j+1}\rightarrow\ldots\rightarrow (u_k=v_k)\rightarrow v_{k+1}$ is a simple cycle in $G$. Now, we reach a contradiction in the same manner as it is reached in the last paragraph of the proof of the previous observation.\qed
\end{proof}

Consider the iteration of Step \ref{step:ite} that corresponds to $v=v_1$ and $u=v_k$. The first observation implies that the condition of Step \ref{step:if1} is false. Next, the second observation implies that the condition of Step \ref{step:if2} is true, and therefore \alg{PolyAlg} accepts.\qed
\end{proof}

\section{Computing the Sets $L$ and $R$}

In this section we observe that the computation of the sets $L$ and $R$ can merely rely on a simple partitioning step. To this end, we need the following definition and known result.

\begin{definition}\label{def:universalSet}
Let ${\cal F}$ be a set of functions $f: \{1,2,\ldots,n\}\rightarrow \{0,1\}$. We say that ${\cal F}$ is an $(n,t)$-universal set if, for every subset $I\subseteq\{1,2,\ldots,n\}$ of size $t$ and a function $f':I\rightarrow\{0,1\}$, there is a function $f\in{\cal F}$ such that, for all $i\in I$, $f(i)=f'(i)$.
\end{definition}

\begin{lemma}[\cite{splitter}]\label{lemma:splitter}
There is a deterministic algorithm that given a pair of integers $(n,t)$, computes in $O^*(2^{t+o(t)})$ time and space an $(n,t)$-universal set ${\cal F}\subseteq 2^{\{1,2,\ldots,n\}}$ of size $O^*(2^{t+o(t)})$.
\end{lemma}

Now, we turn to prove the following simple observations.

\begin{obs}\label{obs:LRDet}
Let $(G=(V,E),k)$ be a instance of {\sc LDC}. Then, there is a deterministic algorithm, \alg{DetLRAlg}, that uses $O^*(4^{k+o(k)})$ time and space, and returns a set $S=\{(L,R): L\subseteq V, R=V\setminus L\}$ of size $O^*(4^{k+o(k)})$ such that the following condition is satisfied.
\begin{itemize}
\item For any simple cycle $v_1\rightarrow v_2\rightarrow\ldots\rightarrow v_t\rightarrow v_1$ of $G$ such that $t\geq 2k$, there exists $(L,R)\in S$ such that $v_1,v_2,\ldots,v_k\in L$ and $v_{k+1},v_{k+2},\ldots,v_{2k}\in R$.
\end{itemize}
\end{obs}

\begin{proof}
\alg{DetLRAlg} arbitrarily orders $V$, and denotes $V=\{v_1,v_2,\ldots,v_{|V|}\}$ accordingly. It obtains an $(|V|,2k)$-universal set ${\cal F}$ by relying on Lemma \ref{lemma:splitter}. Then, it defines $L_f=\{v_i\in V: f(i)=0\}$ and $R_f=V\setminus L$ for each $f\in {\cal F}$, and lets $S=\{(L_f,R_f): f\in{\cal F}\}$. The correctness and running time of the algorithm follow immediately from Definition \ref{def:universalSet} and Lemma \ref{lemma:splitter}.\qed
\end{proof}

\begin{obs}\label{obs:LRRand}
Let $(G=(V,E),k)$ be a instance of {\sc LDC}. Then, there is a randomized algorithm, \alg{RandLRAlg}, with polynomial time and space complexities, that returns a partition $(L,R)$ of $V$. Moreover, if \alg{RandLRAlg} is called $c\cdot 4^k$ times for some $c\geq 1$, and $G$ contains a simple cycle $v_1\rightarrow v_2\rightarrow\ldots\rightarrow v_t\rightarrow v_1$ such that $t\geq 2k$, then with probability at least $(1-e^{-c})$, at least one of the calls returns a pair $(L,R)$ such that $v_1,v_2,\ldots,v_k\in L$ and $v_{k+1},v_{k+2},\ldots,v_{2k}\in R$.
\end{obs}

\begin{proof}
\alg{RandLRAlg} initializes $L$ to be an empty set, and $R$ to be $V$. For each $v\in V$, with probability $\frac{1}{2}$ it removes $v$ from $R$ and inserts $v$ into $L$. Then, it returns the resulting pair $(L,R)$, which is clearly a partition of $V$.

To prove the correctness of \alg{RandLRAlg}, suppose that $G$ contains a simple cycle $v_1\rightarrow v_2\rightarrow\ldots\rightarrow v_t\rightarrow v_1$ such that $t\geq 2k$. Then, the probability that $v_1,v_2,\ldots,v_k\in L$ and $v_{k+1},v_{k+2},\ldots,v_{2k}\in R$ is $(\frac{1}{2})^{2k}=\frac{1}{4^k}$. Now, if \alg{RandLRAlg} is called $c\cdot 4^k$ times, the probability that none of the calls returns a pair $(L,R)$ such that $v_1,v_2,\ldots,v_k\in L$ and $v_{k+1},v_{k+2},\ldots,v_{2k}\in R$ is $(1-\frac{1}{4^k})^{c\cdot4^k}\leq e^{-c}$.\qed
\end{proof}

\section{Solving the {\sc LDC} Problem}

We are now ready to solve {\sc LDC}. The input for our algorithm, \alg{LDCALg}, consists of an instance $(G,k)$ of {\sc LDC}, an algorithm $ALG$ for {\sc $k$-Path}, and an argument $X\in\{det,rand\}$ that specifies whether $ALG$ is deterministic or randomized. \alg{LDCAlg} first determines whether $G$ contains a simple cycle on $\ell$ vertices, for any $\ell\in\{k,k+1,\ldots,2k\}$ by calling $ALG$. If no such cycle is found, \alg{LDCAlg} examines enough pairs $(L,R)$, computed using the algorithm in Observation \ref{obs:LRDet} or \ref{obs:LRRand}, and accepts {\em iff} \alg{PolyAlg} accepts one of the resulting inputs $(G,k,L,R)$. The pseudocode of \alg{LDCALg} is given in Algorithm \ref{alg:ldc}.

\begin{algorithm}[!ht]
\caption{\alg{LDCAlg}($G=(V,E),k,ALG,X$)}\label{alg:ldc}
\begin{algorithmic}[1]
\FOR{$\ell=k,k+1,\ldots,2k$}\label{step:beginSmall}
	\FORALL{$(u,v)\in E$}
		\STATE Use $ALG$ to determine whether $G$ contains a simple path on exactly $\ell$ vertices directed from $v$ to $u$. If the answer is positive, accept.
	\ENDFOR
\ENDFOR\label{step:endSmall}
\IF{$X=det$}
	\STATE Let $S$ be the set returned by \alg{DetLRAlg} (see Observation \ref{obs:LRDet}), ordered arbitrarily. Moreover, let $x=|S|$, and let \alg{PartitionAlg} be a procedure that when called at the $i^{st}$ time, returns the $i^{st}$ pair $(L,R)$ in $S$.
\ELSE
	\STATE Let $x=10\cdot 4^k$, and let \alg{PartitionAlg} be \alg{RandLRAlg} (see Observation \ref{obs:LRRand}).
\ENDIF
\FOR{$i=1,2,\ldots,x$}
	\STATE Call \alg{PartitionAlg} to obtain a pair $(L,R)$.
	\STATE If \alg{PolyAlg}$(G,k,L,R)$ accepts: Accept.
\ENDFOR
\STATE Reject.
\end{algorithmic}
\end{algorithm}

\begin{theorem}
Let $ALG$ be an algorithm that uses $t(G,k)$ time and $s(G,k)$ space, and decides whether $G$ contains a simple path on {\em exactly} $k$ vertices directed from $v$ to $u$ for some given vertices $v,u\in V$.
Then, \alg{LDCAlg} solves {\sc LDC} in deterministic time $O^*(\max\{t(G,2k),4^{k+o(k)}\})$ and $O^*(\max\{s(G,k),4^{k+o(k)}\})$ space (if $ALG$ is deterministic), or in randomized time $O^*(\max\{t(G,2k),4^k\})$ and $O^*(s(G,k))$ space (if $ALG$ is randomized).
\end{theorem}

\begin{proof}
First, observe that the time and space complexities of \alg{LDCAlg} directly follow from the pseudocode, Lemma \ref{lemma:poly} and Observations \ref{obs:LRDet} and \ref{obs:LRRand}. Moreover, by Lemma \ref{lemma:poly} and the correctness of $ALG$, if \alg{LDCAlg} accepts, it is clearly correct  (if $X=rand$, we mean that \alg{LDCAlg} accepts with high probability).\footnote{By iteratively removing edges from $G$, it is easy to see that one can use $ALG$ not only to determine whether $G$ contains a simple path on exactly $\ell$ vertices from $v$ to $u$, but also to return such a path. In this manner, even if $X=rand$, \alg{LCDAlg} can be modified to accept only if $(G,k)$ is a yes-instance.}

Now, to complete the proof, suppose that $(G,k)$ is a yes-instance. If $G$ contains a simple cycle on $\ell$ vertices for some $\ell\in\{k,k+1,\ldots,2k\}$, then one of the calls to $ALG$ accepts, and therefore \alg{LDCAlg} accepts (if $X=rand$, we mean that \alg{LDCAlg} accepts with high probability). Thus, we can next assume that $(G,k)$ seems difficult, and let $C=v_1\rightarrow v_2\rightarrow\ldots\rightarrow v_t\rightarrow v_1$ denote a simple cycle in $G$, where $t>2k$. By Observations \ref{obs:LRDet} and \ref{obs:LRRand}, there is a call to \alg{PartitionAlg} where it returns a pair $(L,R)$ such that $v_1,v_2\ldots,v_k\in L$ and $v_{k+1},v_{k+2},\ldots,v_{2k}\in R$ (in case $X=rand$, we mean that there is such a call with high probability). Then, by Lemma \ref{lemma:poly}, \alg{PolyAlg} accepts, and therefore \alg{LDCAlg} accepts.\qed
\end{proof}

\bibliographystyle{model1-num-names}
\bibliography{references}

\begin{thebibliography}{7}
\expandafter\ifx\csname natexlab\endcsname\relax\def\natexlab#1{#1}\fi
\providecommand{\url}[1]{\texttt{#1}}
\providecommand{\href}[2]{#2}
\providecommand{\path}[1]{#1}
\providecommand{\DOIprefix}{doi:}
\providecommand{\ArXivprefix}{arXiv:}
\providecommand{\URLprefix}{URL: }
\providecommand{\Pubmedprefix}{pmid:}
\providecommand{\doi}[1]{\href{http://dx.doi.org/#1}{\path{#1}}}
\providecommand{\Pubmed}[1]{\href{pmid:#1}{\path{#1}}}
\providecommand{\bibinfo}[2]{#2}
\ifx\xfnm\relax \def\xfnm[#1]{\unskip,\space#1}\fi
\bibitem[{Fomin et~al.(2014)Fomin, Lokshtanov, and Saurabh}]{representative}
\bibinfo{author}{F.~V. Fomin}, \bibinfo{author}{D.~Lokshtanov},
  \bibinfo{author}{S.~Saurabh},
\newblock \bibinfo{title}{Efficient computation of representative sets with
  applications in parameterized and exact agorithms},
\newblock in: \bibinfo{booktitle}{SODA (see also arXiv:1304.4626)},
  \bibinfo{year}{2014}, pp. \bibinfo{pages}{142--151}.
\bibitem[{Gabow and Nie(2008)}]{gabowNie08}
\bibinfo{author}{H.~N. Gabow}, \bibinfo{author}{S.~Nie},
\newblock \bibinfo{title}{Finding a low directed cycle},
\newblock \bibinfo{journal}{ACM Transactions on Algorithms} \bibinfo{volume}{4}
  (\bibinfo{year}{2008}).
\bibitem[{Fomin et~al.(2014)Fomin, Lokshtanov, Panolan, and
  Saurabh}]{productFam}
\bibinfo{author}{F.~V. Fomin}, \bibinfo{author}{D.~Lokshtanov},
  \bibinfo{author}{F.~Panolan}, \bibinfo{author}{S.~Saurabh},
\newblock \bibinfo{title}{Representative sets of product families},
\newblock in: \bibinfo{booktitle}{ESA}, \bibinfo{year}{2014}, pp.
  \bibinfo{pages}{443--454}.
\bibitem[{Shachnai and Zehavi(2014)}]{repFamUniApp}
\bibinfo{author}{H.~Shachnai}, \bibinfo{author}{M.~Zehavi},
\newblock \bibinfo{title}{Representative families: a unified tradeoff-based
  approach},
\newblock in: \bibinfo{booktitle}{ESA}, \bibinfo{year}{2014}, pp.
  \bibinfo{pages}{786--797}.
\bibitem[{Williams(2009)}]{williamsKPath}
\bibinfo{author}{R.~Williams},
\newblock \bibinfo{title}{Finding paths of length $k$ in ${O}^*(2^k)$ time},
\newblock \bibinfo{journal}{Inf. Process. Lett.} \bibinfo{volume}{109}
  (\bibinfo{year}{2009}) \bibinfo{pages}{315--318}.
\bibitem[{Zehavi(2015)}]{mixing}
\bibinfo{author}{M.~Zehavi},
\newblock \bibinfo{title}{Mixing color coding-related techniques},
\newblock in: \bibinfo{booktitle}{ESA}, \bibinfo{year}{2015}.
\bibitem[{Naor et~al.(1995)Naor, Schulman, and Srinivasan}]{splitter}
\bibinfo{author}{M.~Naor}, \bibinfo{author}{J.~L. Schulman},
  \bibinfo{author}{A.~Srinivasan},
\newblock \bibinfo{title}{Splitters and near-optimal derandomization},
\newblock in: \bibinfo{booktitle}{FOCS}, \bibinfo{year}{1995}, pp.
  \bibinfo{pages}{182--191}.

\end{thebibliography}

\end{document}